\documentclass[10pt]{article}
\usepackage{amsthm,amsmath,amssymb}
\usepackage{subfig,sidecap,caption}
\usepackage{paralist}
\usepackage[usenames,dvipsnames]{color}
\usepackage[pdftex,breaklinks,colorlinks,
    citecolor={BlueViolet}, linkcolor={Blue},urlcolor=Maroon]{hyperref}  
\usepackage{tikz,pgfkeys}
\usepackage{tkz-graph}
\usetikzlibrary{decorations.pathmorphing}
\usetikzlibrary{quotes}
\usetikzlibrary{arrows.meta}
\usetikzlibrary{arrows,shapes}
\usetikzlibrary{decorations.pathreplacing, decorations.shapes}
\usepackage{graphicx}
\usepackage{charter,eulervm}%
\usepackage{multirow,booktabs,array}
\usepackage{tabularx}
\usepackage[final,expansion=alltext,protrusion=true]{microtype}

\theoremstyle{plain} 
\newtheorem{theorem}{Theorem}[section]
\newtheorem{lemma}[theorem]{Lemma}
\newtheorem{corollary}[theorem]{Corollary}
\newtheorem{proposition}[theorem]{Proposition}

\newcommand{\hsc}[1]{{\footnotesize\sf\MakeUppercase{#1}}}
\newcommand{\problem}[2]{\hsc{#1} $\rightarrow$ \hsc{#2}}
\newcommand{\lp}[1]{\ensuremath{{\mathtt{lp}(#1)}}}
\newcommand{\rp}[1]{\ensuremath{{\mathtt{rp}(#1)}}}
\newcommand{\comment}[1]{\textbackslash\!\!\textbackslash {\em #1}}
\tikzstyle{class} = [shape=rectangle, rounded corners, draw, align=center, top color=white, bottom color=blue!20]
\tikzstyle{vertex}  = [{fill=blue,circle,draw,inner sep=1pt}]

\title{Vertex Deletion Problems on Chordal Graphs}

\author{Yixin Cao\thanks{Department of Computing, Hong Kong
    Polytechnic University, Hong Kong,
    China. \href{mailto:yixin.cao@polyu.edu.hk} {\tt
      \{yixin.cao, yu.ke, jie.you\}@polyu.edu.hk}}
  \and
  Yuping Ke\addtocounter{footnote}{-1}\footnotemark
  \and
  Yota Otachi\thanks{Faculty of Advanced Science and Technology, Kumamoto University. Kumamoto, Japan. \href{mailto:otachi@cs.kumamoto-u.ac.jp}{otachi@cs.kumamoto-u.ac.jp}}
  \and
  Jie You{\addtocounter{footnote}{-2}\footnotemark}
 }

\date{}

\begin{document}
\maketitle

\begin{abstract}
  Containing many classic optimization problems, the family of vertex deletion problems has an important position in algorithm and complexity study.  The celebrated result of Lewis and Yannakakis gives a complete dichotomy of their complexity.  It however has nothing to say about the case when the input graph is also special.  This paper initiates a systematic study of  vertex deletion problems from one  subclass of chordal graphs to another.  We give polynomial-time algorithms or proofs of NP-completeness for most of the problems.  In particular, we show that the vertex deletion problem from chordal graphs to interval graphs is NP-complete.
\end{abstract}

\section{Introduction}\label{sec:intro}
Generally speaking, a vertex deletion problem asks to transform an input graph to a graph in a certain class by deleting a minimum number of vertices.  Many classic optimization problems belong to the family of vertex deletion problems, and their algorithms and complexity have been intensively studied.  For example, the clique problem and the independent set problem are nothing but the vertex deletion problems to complete graphs and to edgeless graphs respectively.   Most interesting graph properties are \emph{hereditary}: If a graph satisfies this property, then so does every induced subgraph of it.  For all the vertex deletion problems to hereditary graph classes, Lewis and Yannakakis~\cite{lewis-80-node-deletion-np} have settled their complexity once and for all with a dichotomy result: They are either NP-hard or trivial.   Thereafter algorithmic efforts were mostly focused on the nontrivial ones, and the major approaches include  approximation algorithms \cite{lund-93-approximation-maximum-subgraph}, parameterized algorithms \cite{cai-96-hereditary-graph-modification}, and exact algorithms \cite{fomin-16-exact-via-monotone-local-search}.

Chordal graphs make one of the most important graph classes.  Together with many of its subclasses, it has played important roles in the development of structural graph theory.  (We defer their definitions to the next section.)  Many algorithms have been developed for vertex deletion problems to chordal graphs and its subclasses,---most notably (unit) interval graphs, cluster graphs, and split graphs; see, e.g., \cite{fomin-15-large-induced-subgraphs, bliznets-16-max-chordal-interval-subgraphs, cao-16-chordal-editing, cao-15-interval-deletion, cao-17-unit-interval-editing, cao-17-cluster-vertex-deletion, cygan-13-split-vertex-deletion, jansen-16-approximation-and-kernelization-chordal-deletion, agrawal-16-chordal-deletion} for a partial list.
After the long progress of algorithmic achievements, some natural questions arise: What is the complexity of transforming a chordal graph to a (unit) interval graph, a cluster graph, a split graph, or a member of some other subclass of chordal graphs?  It is quite surprising that this type of problems has not been systematically studied, save few concrete results, e.g., the polynomial-time algorithms for the clique problem, the independent set problem, and the feedback vertex set problem (the object class being forests)~\cite{gavril-72-coloring, yannakakis-87-k-colorable-subgraph}.

The same question can be asked for other pair of source and object graph classes.  The most important source classes include planar graphs \cite{garey-76-simplified-problems, garey-77-rectilinear-steiner-tree, fomin-12-f-deletion}, bipartite graphs \cite{yannakakis-81-bipartite-node-deletion}, and degree-bounded graphs~\cite{garey-79}.  As one may expect, with special properties imposed on input graphs, the problems become easier, and some of them may not remain NP-hard.  Unfortunately, a clear-cut answer to them seems very unlikely, since their complexity would depend upon both the source class and the object class.  Indeed, some are trivial (e.g., vertex cover on split graphs), some remain NP-hard (e.g., vertex cover on planar graphs), while some others are in P but can only be solved by very nontrivial polynomial-time algorithms (e.g., vertex cover on bipartite graphs).

Throughout the paper we write the names of graph classes in {small capitals}; e.g., \hsc{chordal}  and \hsc{bipartite} stand for the class of chordal graphs and the class of bipartite graphs respectively.   We use $\cal C$, commonly with subscripts, to denote an unspecified hereditary graph class, and use ${\cal C}_1 \rightarrow {\cal C}_2$ to denote the vertex deletion problem from class ${\cal C}_1$ to class ${\cal C}_2$:
\begin{quote}
 Given a graph $G$ in ${\cal C}_1$, one is asked for a minimum set $V_-\subseteq V(G)$ such that $G - V_-$ is in ${\cal C}_2$.  
\end{quote}
  It is worth noting that ${\cal C}_2$ may or may not be a subclass of ${\cal C}_1$, and when it is not, the problem is equivalent to ${\cal C}_1 \rightarrow {\cal C}_1\cap {\cal C}_2$: Since ${\cal C}_1$ is hereditary, $G - V_-$ is necessarily in ${\cal C}_1$.  For almost all classes $\cal C$, the complexity of problems \problem{planar}{}{$\cal C$} and \problem{bipartite}{}{$\cal C$} has been answered in a systematical manner \cite{lewis-80-node-deletion-np, yannakakis-81-bipartite-node-deletion}, while for most graph classes $\cal C$, the complexity of problem \problem{degree-bounded}{}{$\cal C$} has been satisfactorily determined \cite{garey-79}.

Apart from \hsc{chordal}, we would also consider vertex deletion problems on its subclasses.  Therefore, our purpose in this paper is a focused study on the algorithms and complexity of ${\cal C}_1 \rightarrow {\cal C}_2$ with both ${\cal C}_1$ and ${\cal C}_2$ being subclasses of \hsc{chordal}.  
Since it is generally acknowledged that the study of chordal graphs motivated the theory of perfect graphs \cite{hajnal-58-chordal-graphs, berge-67-some-perfect-graphs}, the importance of chordal graphs merits such a study from the aspect of structural graph theory.
However, our main motivation is from the recent algorithmic progress in vertex deletion problems.  It has come to our attention that to transform a graph to class ${\cal C}_1$, it is frequently convenient to first make it a member of another class ${\cal C}_2$ that contains ${\cal C}_1$ as a proper subclass, followed by an algorithm for the ${\cal C}_2 \rightarrow {\cal C}_1$ problem~\cite{bevern-10-pivd, cao-15-interval-deletion, cao-16-almost-interval-recognition, cao-17-cluster-vertex-deletion}.

There being many subclasses of \hsc{chordal}, the number of problems fitting in our scope is quite prohibitive.  The following simple observations would save us a lot of efforts.

\begin{proposition}\label{lem:general-polynomial}
  Let ${\cal C}_1$ and ${\cal C}_2$ be two graph classes. 
  \begin{enumerate}[(1)]
  \item If the ${\cal C}_1 \rightarrow {\cal C}_2$ problem can be
    solved in polynomial time, then so is ${\cal C} \rightarrow {\cal
      C}_2$ for any subclass ${\cal C}$ of ${\cal C}_1$.
  \item If the ${\cal C}_1 \rightarrow {\cal C}_2$ problem is NP-complete,
    then so is ${\cal C} \rightarrow {\cal C}_2$ for any superclass
    ${\cal C}$ of ${\cal C}_1$.
  \end{enumerate}
\end{proposition}

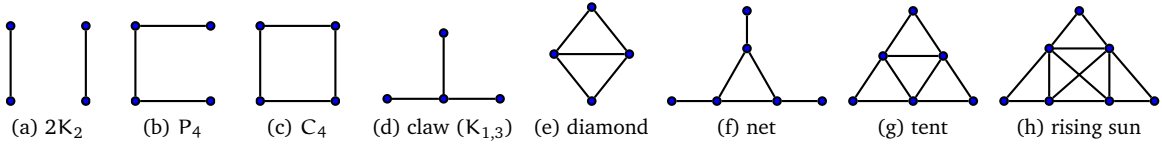
\begin{figure*}[t]
  \centering\footnotesize
  \subfloat[{$2 K_2$}]{\label{fig:2k2}
    \begin{tikzpicture}[auto=left,every node/.style={vertex}, every path/.style={thick},scale=.5]
      \node (a) at (-1,0) {};
      \node (c) at (1,0) {};
      \node (b) at (-1,2) {};
      \node (d) at (1,2) {};
      \draw (a) -- (b) (c) -- (d);
      \node[white] at (1.25,0) {};
      \node[white] at (-1.25,0) {};
    \end{tikzpicture}
  }
  \,
  \subfloat[{$P_4$}]{\label{fig:p4}
    \begin{tikzpicture}[auto=left,every node/.style={vertex}, every path/.style={thick},scale=.5]
      \node (a) at (-1,0) {};
      \node (c) at (1,0) {};
      \node (b) at (-1,2) {};
      \node (d) at (1,2) {};
      \draw (d) -- (b) -- (a) -- (c);
      \node[white] at (1.25,0) {};
      \node[white] at (-1.25,0) {};
    \end{tikzpicture}
  }
  \,
  \subfloat[{$C_4$}]{\label{fig:c4}
    \begin{tikzpicture}[every node/.style={vertex}, every path/.style={thick}, scale=.5]
      \node (a) at (-1,0) {};
      \node (c) at (1,0) {};
      \node (b) at (-1,2) {};
      \node (d) at (1,2) {};
      \draw (a) -- (b) -- (d) -- (c) -- (a);
      \node[white] at (1.25,0) {};
      \node[white] at (-1.25,0) {};
    \end{tikzpicture}
  }
  \,
  \subfloat[claw ($K_{1,3}$)]{\label{fig:claw}
    \begin{tikzpicture}[every path/.style={thick}, scale=.25]
    \node [vertex] (a1) at (-3., 0) {};
    \node [vertex] (v) at (0, 0) {};
    \node [vertex] (b1) at (3., 0) {};
    \node [vertex] (c) at (0,3.5) {};
    \node at (-3.5, 0) {};
    \node at (3.5, 0) {};
    \draw[] (a1) -- (v) -- (b1);
    \draw[] (v) -- (c);
    \end{tikzpicture}
  }
  \subfloat[{diamond}]{\label{fig:diamond}
    \begin{tikzpicture}[every path/.style={thick}, scale=.25]
    \node [vertex] (a1) at (-2, 0) {};
    \node [vertex] (v) at (0, -2.5) {};
    \node [vertex] (b1) at (2, 0) {};
    \node [vertex] (c) at (0, 2.5) {};
    \node at (3, 0) {}; //position adjustment
    \node at (-3, 0) {}; //position adjustment
    \draw (b1) -- (c) -- (a1) -- (v) -- (b1) -- (a1);
  \end{tikzpicture}
  }
  \subfloat[net]{\label{fig:net}
    \begin{tikzpicture}[every path/.style={thick}, scale=.2]
    \node [vertex] (s) at (0,6) {};
    \node [vertex] (a) at (-5,0) {};
    \node [vertex] (a1) at (-2,0) {};
    \node [vertex] (b1) at (2,0) {};
    \node [vertex] (b) at (5,0) {};
    \node [vertex] (c) at (0,3.5) {};
    \draw[] (a) -- (a1) -- (b1) -- (b);
    \draw[] (c) -- (s);
    \draw[] (a1) -- (c) -- (b1);
    \end{tikzpicture}
  }
  \,
  \subfloat[tent]{\label{fig:tent}
    \begin{tikzpicture}[every path/.style={thick}, scale=.2]
    \node [vertex] (s) at (0,6) {};
    \node [vertex] (a) at (-4,0) {};
    \node [vertex] (a1) at (0, 0) {};
    \node [vertex] (b) at (4,0) {};
    \node [vertex] (c1) at (-2,3) {};
    \node [vertex] (c2) at (2,3) {};
    \draw[] (a) -- (a1) -- (b) -- (c2) -- (s) -- (c1) -- (a);
    \draw[] (c1) -- (c2) -- (a1) -- (c1);
    \end{tikzpicture}
  }
  \,
  \subfloat[rising sun]{\label{fig:rising-sun}
    \begin{tikzpicture}[every path/.style={thick}, scale=.2]
    \node [vertex] (s) at (0,6) {};
    \node [vertex] (a) at (-5, 0) {};
    \node [vertex] (a1) at (-2,0) {};
    \node [vertex] (b1) at (2,0) {};
    \node [vertex] (b) at (5,0) {};
    \node [vertex] (c1) at (-2,3.5) {};
    \node [vertex] (c2) at (2,3.5) {};
    \draw[] (a) -- (a1) -- (b1) -- (b) -- (c2) -- (s) -- (c1) -- (a);
    \draw[] (c1) -- (c2) -- (a1) -- (c1) -- (b1) -- (c2);
    \end{tikzpicture}
  }
  \caption{Small subgraphs.}
  \label{fig:small-graphs}
\end{figure*}

For example, the majority of our hardness results for problems \problem{chordal}{}{$\cal C$} are obtained by proving the hardness of \problem{split}{}{$\cal C$}.  Indeed, this is very natural as in literature, most (NP-)hardness of problems on chordal graphs is proved on split graphs, e.g., dominating set~\cite{bertossi-84-domination-split-bipartite}, Hamiltonian path~\cite{muller-96-hamiltonian-chordal-bipartite}, and maximum cut~\cite{bodlaender-00-maximum-cut}.  The most famous exception is probably the pathwidth problem, which can be solved in polynomial time on split graphs but becomes NP-complete on chordal graphs \cite{gustedt-93-pathwidth-chordal}.  No problem like this surfaces during our study, though we do have the following hardness result proved directly on chordal graphs, for which we have no conclusion on split graphs.
\begin{theorem}\label{thm:biconnected}
Let $F$ be a biconnected chordal graph.  If $F$ is not complete, then the \problem{chordal}{$F$-free} problem is NP-complete.
\end{theorem}

Another simple observation of common use to us is about complement graph classes.  The \emph{complement} $\overline G$ of graph $G$ is defined on the same vertex set $V(G)$, where a pair of distinct vertices $u$ and $v$ is adjacent in $\overline G$ if $u v \not\in E(G)$.  It is easy to see that the complement of $\overline G$ is $G$.  In Figure~\ref{fig:small-graphs}, for example,  the net and the tent are the complements of each other.  The \emph{complement} of a graph class $\cal C$, denoted by $\overline{\cal C}$, comprises all graphs whose complements are in $\cal C$; e.g., the complement of \hsc{complete split} is \hsc{$\{2K_2, P_3\}$-free}.  A graph class $\cal C$ is \emph{self-complementary} if it is its own complement, i.e., a graph $G \in \cal C$ if and only if $\overline G \in \cal C$.  For example, both \hsc{split} and \hsc{threshold} are self-complementary.\footnote{We should not confuse the self-complementary property of graph classes and the self-complementary property of graphs---a graph is \emph{self-complementary} if it is isomorphic to its complement.  For example, the statement ``\textit{threshold graphs are self-complementary}'' is incorrect, because most threshold graphs are not isomorphic to their complements, though the later are necessarily threshold graphs.}
As usual, $n$ denotes the number of vertices in the input graph.  Note that we need an $n^2$ item because it takes $O(n^2)$ time to compute the complement of a graph.

\begin{proposition}\label{lem:complement}
  Let ${\cal C}_1$ and ${\cal C}_2$ be two graph classes.  If the ${\cal C}_1 \rightarrow {\cal C}_2$ problem can be solved in $f(n)$ time, then the $\overline {\cal C}_1 \rightarrow \overline{\cal C}_2$ problem can be solved in $O(f(n) + n^2)$ time.
\end{proposition}
We are now ready to summarize our results (besides Theorem~\ref{thm:biconnected}) in Figure~\ref{fig:overview}.
\begin{figure}[h]
  \centering\small
  \begin{tikzpicture}[every path/.style={thick}, scale = 1.2]
    \node[class] (chordal) at (2,6.5) {\hsc{chordal}};
    \node[class] (split) at (-0.5,5) {\hsc{split}};
    \node[class] (interval) at (4.5,5) {\hsc{interval}};
    \node[class] (threshold) at (-0.5,1.5) {\hsc{threshold}};
    \node[class] (tpg) at (1,3.5) {\hsc{trivially}\\\hsc{perfect}};
    \node[class] (uig) at (5, 3) {\hsc{unit interval}};
    \node[class] (block) at (2.65, 3.5) {\hsc{block}};
    \node[class] (cluster) at (3,1.5) {\hsc{cluster}};
    \node[class] (csplit) at (-1.5,0) {\hsc{complete split}};
    \node[class] (bottom) at (1.5,0) {\hsc{$\{2K_2, P_3\}$-free}};
    \node[class] (cochain) at (5,0) {\hsc{co-chain}};
    \draw[cyan] (csplit) -- (threshold);
    \draw[violet] (threshold) -- (split) node [midway,fill=white] {\scriptsize NPC};
    \draw[cyan] (split) -- (chordal) node [midway,fill=white] {\scriptsize P}; 
    \draw[violet] (interval) -- (chordal); 
    \draw[violet, dashed, -latex] (split) -- (interval) node [pos=.8,fill=white] {\scriptsize NPC};
    \draw (uig) -- (interval); 
    \draw[cyan] (cluster) -- (uig);
    \draw[cyan] (cluster) -- (block);
    \draw[violet] (block.north) -- (chordal); 
    \draw[violet, dashed, -latex] (split) -- (block.north) node [pos=.3,fill=white] {\scriptsize NPC};
    \draw[cyan] (bottom) -- (cluster);
    \draw[cyan] (bottom) edge (threshold);
    \draw[cyan, dashed, -latex] (split) -- (bottom) node [pos=.85, fill=white] {\scriptsize P};
    \draw[cyan, dashed, -latex] (split) -- (csplit) node [pos=.82, fill=white] {\scriptsize P};
    \draw[cyan, dashed, -latex] (interval) -- (cluster) node [pos=.2, fill=white] {\scriptsize P};   
    \draw (threshold) -- (tpg) -- (interval);
    \draw[cyan] (cochain) -- (uig);
    \draw[cyan, dashed, -latex] (chordal) -- (cochain) node [pos=.85, fill=white] {\scriptsize P};
  \end{tikzpicture}
  \caption{A summary of major graph classes studied by this paper and our results.  Two classes are connected by a solid edge when the lower one is a subclass of the higher one.  A directed dashed edge from ${\cal C}_1$ to ${\cal C}_2$ is used when ${\cal C}_2$ is not an immediate subclass of ${\cal C}_1$.  We omit here results implied by Proposition~\ref{lem:general-polynomial}; e.g., {${\cal C}$}\problem{}{co-chain} is in P for all $\cal C$, and \problem{chordal}{}{${\cal C}$} is NP-complete when $\cal C$ is \hsc{threshold}, \hsc{block}, or \hsc{interval}.  The cyan, violet, and black edges indicate that the complexity of the representing problems is in P, NP-complete, and unknown, respectively.}
  \label{fig:overview}
\end{figure}
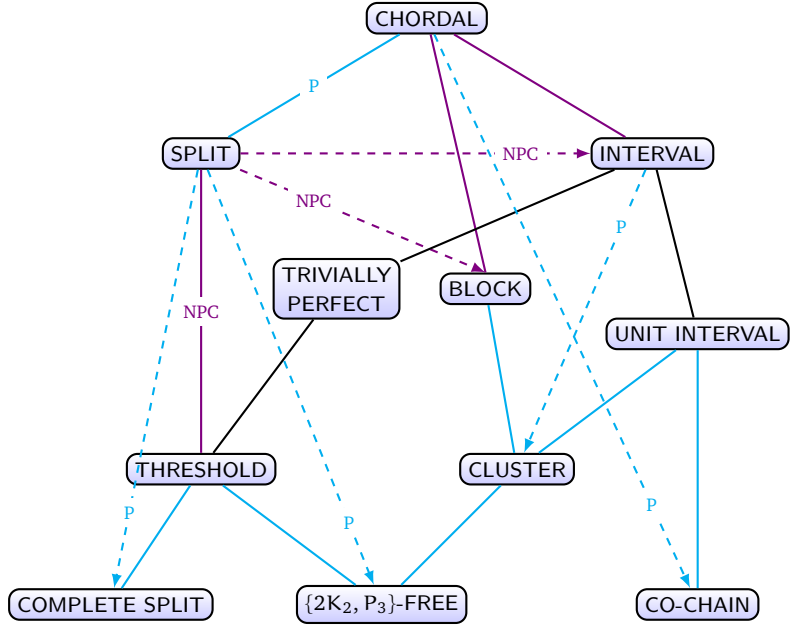

Unfortunately, we have to leave the complexity of some problems open, particularly \problem{chordal}{cluster}, 
 \problem{chordal}{unit interval}, and \problem{interval}{unit interval}.
Our final remarks are on the approximation algorithms, for which we are concerned with those not shown to be in P.  All of them have constant-ratio approximations, which follow from either \cite{cao-16-almost-interval-recognition, cao-17-unit-interval-editing} or the general observation of Lund and Yannakakis~\cite{lund-93-approximation-maximum-subgraph}.  On the other hand, none of the NP-complete problems admits a polynomial-time approximation scheme.

\section{Preliminaries}

All graphs discussed in this paper are undirected and simple.  A graph
$G$ is given by its vertex set $V(G)$ and edge set $E(G)$, whose
cardinalities will be denoted by $n$ and $m$ respectively.  
For a subset $X\subseteq V(G)$, denote by $G[X]$ the subgraph induced by $X$, and by $G - X$ the subgraph $G[V(G)\setminus X]$; we use $E(X)$ as a shorthand for $E(G[X])$, i.e., all edges among vertices in $X$.  For a subset $E_-\subseteq E(G)$ of edges, we use $G - E_-$ to denote the subgraph with vertex set $V(G)$ and edge set $E(G)\setminus E_-$.  We write $G - v$ and $G - e$ instead of $G - \{v\}$ and $G - \{e\}$ for $v\in V(G)$ and $e\in E(G)$ respectively.

For $\ell\ge 2$, we use $P_\ell$, $K_\ell$, and $I_\ell$ to denote an induced path, a clique, and an independent set, respectively, on $\ell$ vertices.  For $\ell\ge 4$, we use $C_\ell$ to denote an induced cycle on $\ell$ vertices; such a cycle is also called a \emph{hole}.  
  Some small graphs that will be used in this paper are depicted in Figure~\ref{fig:small-graphs}.  Note that $C_4$ and $2 K_2$ are complements to each other, while the complements of $P_4$ and $C_5$ are themselves.

We say that a graph $G$ \emph{contains} a subgraph $F$ if $F$ is isomorphic to some induced subgraph of $G$.  A graph is {\em $F$-free} if it does not contain $F$; for a set $\cal F$ of graphs, a graph $G$ is {\em $\cal F$-free} if it is {\em $F$-free} for every $F\in \cal F$.  Each set $\cal F$ defines a hereditary graph class, and every hereditary graph class can be defined as such; in other words, for any hereditary graph class $\cal C$, there is a (possibly infinite) set $\cal F$ of subgraphs such that a graph $G\in \cal C$ if and only if it is $\cal F$-free.
  Each graph $F$ in $\cal F$ is usually assumed to be minimal, in the sense that $F$ is not in $\cal C$ but every proper induced subgraph of $F$ is; they are called the \emph{minimal obstructions} of $\cal C$.    One should note that a minimal obstruction of a graph class may not be a minimal obstruction of its subclass; e.g., the minimal obstruction $C_5$ of \hsc{split} is not a minimal obstruction of \hsc{threshold}, because $C_5$ contains the non-threshold graph $P_4$ as a proper induced subgraph.
  
  The vertex deletion problem with object class $\cal C$ can also be defined as finding a maximum subgraph in the class $\cal C$.  For example, both vertex cover and independent set refer to the vertex deletion problem to the class \hsc{edgeless}, which is exactly the $K_2$-free graphs.  Although these formulations may behave different with respect to approximation, they are the same for our purpose.  We may use both formulations interchangeably, dependent on which is more convenient in the context.  Yet another way to view the vertex deletion problem toward property $\cal F$-free is to find a minimum set of vertices from a graph to hit all its induced subgraphs in $\cal F$.

We now define the graph classes we are going to study.  For the convenience of the reader, we collect the obstructions of all the graph classes and their containment relationships in Figure~\ref{fig:classes-containment} of the appendix.  Although the containment relationships of all the graph classes to be studied can be readily checked with their obstruction characterizations, sometimes it would be far more informative and inspiring if we look at them from the lens of the definitions and/or geometric representations of these graph classes.

A graph is \emph{chordal} if every cycle of length larger than three has a chord, i.e., an edge between two non-consecutive vertices of the cycle. 
A graph is an \emph{interval graph} if its vertices can be assigned to intervals on the real line such that there is an edge between two vertices if and only if their corresponding intervals intersect, and a \emph{unit interval graph} if all the intervals have the same length.
A graph $G$ is a \emph{trivially perfect graph }if for every induced subgraph of $G$, the size of the largest independent set is equivalent to the number of all maximal cliques \cite{golumbic-78-trivially-perfect}. 
Chordal graphs are precisely graphs that are intersection graphs of subtrees of a tree, while interval graphs are  intersection graphs of sub-paths of a path.  Therefore, \hsc{interval} $\subset$ \hsc{chordal}.  A trivially perfect graph can be represented by a set of \emph{non-overlapping} intervals; in other words, if two intervals intersect, then one is contained in the other.  Therefore, \hsc{trivially perfect} $\subset$  \hsc{interval}.

A graph is a \emph{cluster graph} if every component is a clique.  A graph is a \emph{block graph} if the deletion of all {cut vertices} leaves a cluster graph.  It is known that a graph is $\{2 K_2,  P_3\}$-free if it is a cluster graph of which at most one clique is nontrivial, i.e.,  having more than one vertex.  It is immediate from their definitions that \hsc{$\{2 K_2,  P_3\}$-free} $\subset$ \hsc{cluster} $\subset$ \hsc{block}.  Moreover, block graphs are precisely those chordal graph of which any two maximal cliques share at most one vertex.

A graph is a \emph{split graph} if its vertices can be partitioned into a clique $C$ and an independent set $I$, and a \emph{complete split graph} if every vertex in $C$ is adjacent to all vertices in $I$; we use $C\uplus I$ to denote the split partition.  Note that either of the two sets may be empty.
A graph $G$ is a \emph{threshold graph} if there is a real number t, the so-called \emph{threshold}, and an assignment $f: V(G)\to \mathbb{R}$ such that $u v\in E(G)$ if and only if $f(u) + f(v) \ge t$ \cite{chvatal-77-inequalities-in-ip}.  It is easy to verify that \hsc{complete split} $\subset$ \hsc{threshold} $\subset$ \hsc{split}: The first can be witnessed by $t = 1$ and assignment $f(v) = 1$ if $v\in C$ and $0$ otherwise; and the second by the clique partition $\{v: f(v) \ge t / 2\}\uplus \{v: f(v) < t / 2\}$.  Further, if we order the vertices in the independent set $I$ of a threshold graph such that 
\[
f(v_1) \le \cdots \le f(v_{|I|}) < t/2,
\]
then 
\[
N(v_1) \subseteq \cdots \subseteq N(v_{|I|}).
\]
Likewise, there is an ordering of vertices $u_1$, $\ldots$, $u_{|C|}$ in $C$ such that $N[u_1] \subseteq \cdots \subseteq N[u_{|C|}]$.

The reader may have noticed the striking resemblance between {split} graphs and {bipartite} graphs.  Indeed, if we add edges to make one side of a bipartite graph into a clique, we end with a split graph; or equivalently, given a split graph $G$ with split partition $C\uplus I$, the subgraph $G - E(C)$ is bipartite.  Clearly, $G - E(C)$ is a complete bipartite graph if and only if $G$  is a complete split graph.  
If $G$ is a threshold graph, then $G - E(C)$ is a \emph{chain graph} \cite{yannakakis-81-bipartite-node-deletion, yannakakis-81-minimum-fill-in}.  Finally,  \hsc{co-chain} denotes the complement of \hsc{chain}.

Recall that Yannakakis~\cite{yannakakis-81-bipartite-node-deletion} has given a dichotomy on the vertex deletion problem from bipartite graphs.  Inspired by this and the aforementioned connection between bipartite graphs and split graphs, a natural attempt at problems \problem{split}{}{$\cal C$} would be reducing them to the corresponding problem on bipartite graphs (for algorithms) or the other way (for hardness results).  This approach however turns out to be less straightforward as one may expect.

The first trouble is that a split graph can have many different split partitions, and thus can be mapped to many different bipartite graphs.  For instance, a naive reduction for the \problem{split}{complete split} problem is to the \problem{bipartite}{complete bipartite} problem, which can be solved in polynomial time.\footnote{We can find a maximum complete bipartite subgraph from a bipartite graph as follows.  We find a maximum independent set of $G$ and a maximum independent set of its {bipartite complement} (i.e., after taking its complement, we discard all edges among the two parts, so the resulting graph remains bipartite with the same partition), and then return the larger of them \cite{yannakakis-81-bipartite-node-deletion}.}  
As shown in Figure~\ref{fig:split-complete-split}, however, this reduction may end with a suboptimal solution.  Some remarks on this example are worthwhile.  The input graph in Figure~\ref{fig:split-complete-split} has a unique split partition.  However, $G - v_3$, which is the unique optimal solution, has four split partitions, of which only one is complete.  As we will see in the next section, this problem can still be solved efficiently by noticing that a split graph can have only a polynomial number of different split partitions, and all of them are very \emph{similar}.

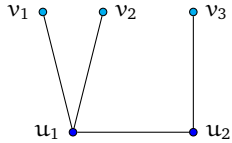
\begin{figure*}[h]
  \centering\small
    \begin{tikzpicture}[scale=.8]
      \node [vertex, "$u_1$" left] (b) at (-0.5,0) {};
      \node [vertex, "$u_2$" right] (c) at (1.5,0) {};
      \node [vertex, "$v_1$" left, fill=cyan] (v1) at (-1,2) {};
      \node [vertex, "$v_2$" right, fill=cyan] (v2) at (0,2) {};
      \node [vertex, "$v_3$" right, fill=cyan] (v3) at (1.5,2) {};
      \draw (b) -- (c);
      \draw (c) -- (v3) ;
      \draw (v1) -- (b) -- (v2);
    \end{tikzpicture}
  \caption{Given is a split graph $G$.  One only needs to delete vertex $v_3$ from $G$ to make it a complete split graph.  However, if we consider the bipartite graph $G - u_1 u_2$, its maximum complete bipartite subgraph has only three vertices.}
  \label{fig:split-complete-split}
\end{figure*}

The situation becomes even more gloomy when we consider the transformation from bipartite graphs to split graphs.  A bipartite graph can have an exponential number of bipartitions, and may be mapped to the same number of distinct split graphs.  Consider, for example, an attempt to find a reduction from the \problem{split}{diamond-free} problem to problem \problem{bipartite}{}{$\cal C$} for some subclass $\cal C$ of \hsc{bipartite}.  A diamond-free split graph admits a split partition $C\uplus I$ such that each vertex in $I$ has degree at most one.  A natural candidate for $\cal C$ is the disjoint union of stars, for which the \problem{bipartite}{}{$\cal C$} problem is known to be NP-complete~\cite{yannakakis-81-bipartite-node-deletion}.  However, the naive reduction would not work: Given a bipartite graph that is a disjoint union of starts, if we take a wrong bipartition and add edges to make it a split graph, we may introduce many diamonds.  As shown in Figure~\ref{fig:bipartite-stars}, even connectedness, which imposes a unique bipartition, would not save us here.

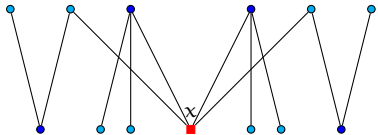
\begin{figure*}[h]
  \centering\footnotesize
    \begin{tikzpicture}[scale=.8]
      \begin{scope}[xshift=-2.5cm]
        \node [vertex] (b) at (-0.5,0) {};
        \node [vertex, fill=cyan] (v1) at (-1,2) {};
        \node [vertex, fill=cyan] (u1) at (0,2) {};
        \draw (v1) -- (b) -- (u1);
      \end{scope}
      \begin{scope}[xshift=-2cm, yshift=2cm, rotate=180]
        \node [vertex] (u2) at (-0.5,0) {};
        \node [vertex, fill=cyan] (v1) at (-0.5,2) {};
        \node [vertex, fill=cyan] (v2) at (0,2) {};
        \draw (v1) -- (u2) -- (v2);
      \end{scope}
      \begin{scope}[yshift= 2cm, rotate = 180]
        \node [vertex] (u3) at (-0.5,0) {};
        \node [vertex, fill=cyan] (v1) at (-1,2) {};
        \node [vertex, fill=cyan] (v2) at (-0.5,2) {};
        \draw (v1) -- (u3) -- (v2);
      \end{scope}
      \begin{scope}[xshift= 2.5cm]
        \node [vertex] (b) at (-0.5,0) {};
        \node [vertex, fill=cyan] (u4) at (-1,2) {};
        \node [vertex, fill=cyan] (v2) at (0,2) {};
        \draw (u4) -- (b) -- (v2);
      \end{scope}
      \node[fill=red, red, draw,inner sep=1.5pt, "$x$"] (v3) at (-0.5, 0) {};
      \draw (u1) -- (v3) -- (u2) (u4) -- (v3) -- (u3);
    \end{tikzpicture}
  \caption{Given is a bipartite graph $G$.  Deleting the vertex $x$ from it leaves a disjoint union of stars.  However, the graph has only two bipartitions, and from the split graphs decided by either of them, we need to delete at least two vertices to make it diamond-free.}
  \label{fig:bipartite-stars}
\end{figure*}

\section{Algorithmic results}

This section gives the polynomial-time algorithms.  Our focus would be laid on the use of structural properties, and if possible, we would present the simplest algorithms without elaborating on the implementation details.  These problems may have more efficient algorithms, and with more complex data structures and algorithmic finesses, some of them may even be solved in linear time.  

Our first two results are on split graphs, for which we need to put split partitions under scrutiny.  Let $C\uplus I$ be a split partition of a split graph $G$.  
  If some vertex in $I$ is completely adjacent to $C$, then we can move such a vertex $v$ to $C$ to make another split partition $C' = C\cup\{v\}$ and $I' = I\setminus\{v\}$.  Note that the vertex $v$ may not be unique, and the resulting graphs by moving them would be isomorphic.  Moreover, after such a move, no vertex of $I'$ can be completely adjacent to $C'$.  The following proposition fully characterizes split graphs with more than one different split partition.

\begin{proposition}\label{lem:split-partitions}
Let $G$ be a split graph with at least two split partitions, and let $C\uplus I$ and $C'\uplus I'$ be two different split partitions of $G$.  
\begin{enumerate}[(i)]
\item The difference between $|C|$ and $|C'|$ is at most  $1$.
\item If $|C| = |C'| + 1$, then $C$ is a maximum clique, and $I'$ is a maximum independent set of $G$; moreover, $C'\subset C$.
\item If $|C| = |C'|$, then  $G - E(C)$ and $G - E(C')$ are isomorphic.
\end{enumerate}
\end{proposition}

As a result, a split graph has either one or two essentially distinct split partitions.  On the other hand, of all split partitions of a complete bipartite graph, only one, whose independent set is the largest, satisfies the definition of complete bipartite graphs, and we will exclusively refer to it when we are discussing a complete split graph.

Let $G$ be a split graph with split partition $C\uplus I$ and let  $\underline G$ be a $\{2 K_2, P_3\}$-free subgraph of $G$.  If $\underline G$ has edges, all of them must be in the same nontrivial clique.  At most one vertex of this clique can be from $I$; therefore, all other vertices of $I$ either are deleted or become isolated in $\underline G$.  In other words, for each other vertex $v$ in $I$, either $v$ or all its neighbors have to be deleted.
\begin{figure}[h!]
  \centering
  \tikz\path (0,0) node[draw, text width=.8\textwidth, rectangle, rounded corners, inner xsep=20pt, inner ysep=10pt]{
    \begin{minipage}[t!]{\textwidth}
  {\sc Input}: a split graph $G$ on split partition $C\uplus I$.
  \\
  {\sc Output}: a minimum set $V_-\subseteq V(G)$ such that $G - V_-$ is {$\{2 K_2, P_3\}$-free}.
  \begin{tabbing}
    Aaa\=aaA\=Aaa\=MMMMMMAAAAAAAAAAAAAAAAAAAAAAAAA\=A \kill
    0.\> ${\cal S}\leftarrow \emptyset$;
    \\
    1.\> build a bipartite graph $G'$ by removing all edges among $C$ from $G$;
    \\
    2.\> find a minimum vertex cover of $G'$, and add it to ${\cal S}$;
    \\
    3.\> {\bf  for each} $v\in I$ {\bf do} 
    \\
    \>\> find a minimum vertex cover $X$ of $G' - (C\setminus N(v)) - v$; 
    \\
    \>\> add $X \cup (C\setminus N(v))$ to $\cal S$;
    \\
    4.\> {\bf return} a set in $\cal S$ with the minimum cardinality.
  \end{tabbing}  
    \end{minipage}
  };
  \caption{Algorithm for \problem{split}{$\{2 K_2, P_3\}$-free}.}
  \label{fig:split-cluster}
\end{figure}

\begin{theorem}\label{thm:split-cluster}
  The \problem{split}{$\{2 K_2, P_3\}$-free} problem is in P.
\end{theorem}
\begin{proof}
  Let $G$ be the input graph to the \problem{split}{$\{2 K_2, P_3\}$-free} problem and let $C\uplus I$ be a split partition of $G$.  We use the algorithm in Figure~\ref{fig:split-cluster} to find a minimum solution to $G$.  To argue its correctness, we show that (i) every set in $\cal S$, added in step~2 or 3, is a solution to $G$, and (ii) at least one of them is minimum.  For (i), it is easy to verify that any vertex cover of $G' = G - E(C)$ is a solution: There is no edge between $C$ and $I$ after its deletion.  The situation in step~3 is similar; note that $N[v] \uplus (I\setminus \{v\})$ is a split partition of $G - \big( C\setminus N(v) \big)$.

  Let $V_-$ be a minimum solution to $G$.  In the first case, every vertex $v\in I\setminus V_-$ is isolated in $G - V_-$.  In other words, $V_-$ contains a vertex cover of $G' = G - E(C)$, and then the solution found by step~2 is already the minimum.  Henceforth we assume that there exists a vertex $v\in I\setminus V_-$ such that $N(v)\not\subseteq V_-$.  Since any vertex $u\in N(v)$ and $w\in C\setminus N(v)$ induce a $P_3$ with $v$, in this case all vertices in $C\setminus N(v)$ must be in $V_-$.    Note that the vertex $v$ is unique: If two vertices in $I\setminus V_-$ have neighbors in $C\setminus V_-$, then they are in a non-clique component.  Therefore, after removing $C\setminus N(v)$ and $v$ from the graph, it reduces to the first case.  This justifies step~3.
 
The algorithm makes $O(n)$ calls to an algorithm for the bipartite vertex cover problem, each  taking $O(m\sqrt{n})$ time, and hence the whole algorithm runs in $O(m n \sqrt{n})$ time.
\end{proof}

Noting that \hsc{split} $\cap$ \hsc{cluster} is precisely \hsc{$\{2 K_2, P_3\}$-free}, we can apply the algorithm of Theorem~\ref{thm:split-cluster} to the \problem{split}{cluster} problem.  Moreover, since \hsc{split} is self-complementary, while the complement of  \hsc{$\{2 K_2, P_3\}$-free} is \hsc{complete split}, it follows from Proposition~\ref{lem:complement} that the \problem{split}{complete split} problem is also in P.  
\begin{corollary}\label{lem:split-complete-split}
  Problems \problem{split}{cluster} and \problem{split}{complete split} are in $P$.  
\end{corollary}

A similar observation as that of the proof Theorem~\ref{thm:split-cluster} can be used to solve the \problem{split}{unit interval} problem.  We start from a simple property of connected graphs in \hsc{split} $\cap$ \hsc{unit interval}.

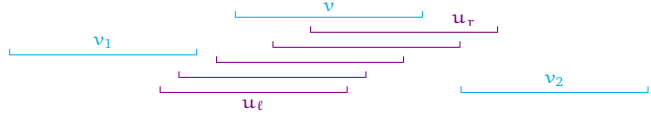
\begin{figure}[h]
  \centering\scriptsize
  \begin{tikzpicture}
    \draw[violet, {|[right]}-{|[left]}]  (2, 0)  to ["$u_\ell$" below] (4.5, 0);
    \draw[violet, {|[right]}-{|[left]}]  (4, 0.8) to ["$u_r$" xshift=8mm] (6.5, 0.8);
    \draw[violet, {|[right]}-{|[left]}]  (2.25, 0.2)  to (4.75,0.2);
    \draw[violet, {|[right]}-{|[left]}]  (2.75,0.4)  to (5.25,0.4);
    \draw[violet, {|[right]}-{|[left]}]  (3.5,0.6)  to (6,0.6);
    \draw[cyan, {|[right]}-{|[left]}]  (0,0.5) to ["$v_1$"] (2.5,0.5);
    \draw[cyan, {|[right]}-{|[left]}]  (6,0) to ["$v_2$"] (8.5,0);
    \draw[cyan, {|[right]}-{|[left]}]  (3,1) to ["$v$"] (5.5,1);
  \end{tikzpicture}
  \caption{A connected split graph with split partition $C\uplus I$ that is also a unit interval graph. Violet intervals are for vertices in $C$ and cyan for $I$.  Note that the vertex $v$ from $I$ is completely adjacent to $C$.} 
  \label{fig:split-and-uig}
\end{figure}
\begin{proposition}\label{lem:split+uig}
  Let $G$ be a connected split graph and let  $C\uplus I$ be a split partition of $G$.  If $G$ is a unit interval graph, then $|I| \le 3$, and the equality holds only when there is a vertex $v\in I$ adjacent to all vertices in $C$.  
\end{proposition}
\begin{proof}
  We prove $|I| \le 2$ if $C$ is a maximum clique of $G$, and then the proposition follows from Proposition~\ref{lem:split-partitions}(i).
  Let $u_\ell$ and $u_r$ be the vertices in $C$ with respectively the leftmost and rightmost intervals.  Suppose for contradiction $|I| > 2$.  Let $v_1$ and $v_2$ be the vertices in $I$ with respectively the leftmost and rightmost intervals.  Then $\lp{u_\ell} < \rp{v_1} < \lp{u_r} < \rp{u_\ell} < \lp{v_2} < \rp{u_r}$, where the second and the fourth inequalities follow from that $C$ is a maximum clique, and the others from the selections of the four vertices.  Since $G$ is connected, the interval for any other vertex $v$ in $I\setminus \{v_1, v_2\}$, which is nonempty, has to lie in ($\rp{v_1}, \lp{v_2}$).  But then it has to contain $[\lp{u_r}, \rp{u_\ell}]$, and $\{v\}\cup C$ is a clique, contradicting that $C$ is a maximum clique of $G$.
\end{proof}

Similar as Theorem~\ref{thm:split-cluster}, our algorithm for \problem{split}{unit interval} separates into two cases, based on whether there is a vertex of $I\setminus V_-$ adjacent to all vertices in $C\setminus V_-$.
\begin{figure}[h!]
  \centering
  \tikz\path (0,0) node[draw, text width=.8\textwidth, rectangle, rounded corners, inner xsep=20pt, inner ysep=10pt]{
    \begin{minipage}[t!]{\textwidth}
  {\sc Input}: a split graph $G$ on split partition $C\uplus I$.
  \\
  {\sc Output}: a minimum set $V_-\subseteq V(G)$ such that $G - V_-$ is a unit interval graph.
  \begin{tabbing}
    Aaa\=aaA\=Aaa\=MMMMMMAAAAAAAAAAAAAAAAAAAAAAAAA\=A \kill
    0.\> ${\cal S}\leftarrow \emptyset$;
    \\
    1.\> solve the \problem{split}{$\{2K_2, P_3\}$-free} problem on $G$; add the solution to ${\cal S}$;
    \\
    \> \comment{ \bf case 1:}
    \\
    2.\> {\bf for each} $v\in I$ {\bf do}
    \\
    \>\>
    find a minimum vertex cover of $G - v - E(C)$, and add it to ${\cal S}$;
    \\
    3.\> {\bf for each} $v_1, v_2\in I$ {\bf do}
    \\
    3.1.\>\> $G' \leftarrow G - \{v_1, v_2\} - E(C)$;
    \\
    3.2.\>\> find a minimum vertex cover of $G' - N(v_1)\cap N(v_2)$,
    \\
    \>\> and add its union with $N(v_1)\cap N(v_2)$ to ${\cal S}$;
    \\
    3.3.\>\> find a minimum vertex cover of $G' - C\setminus \big( N(v_1)\cup N(v_2) \big)$,
    \\
    \>\> and add its union with $C\setminus \big( N(v_1)\cup N(v_2) \big)$ to ${\cal S}$;
    \\
    \> \comment{ \bf case 2:}
    \\
    4.\> {\bf for each} $v\in I$ {\bf do} 
    \\
    \>\> $G'' \leftarrow G - \big(C\setminus N(v) \big)$ with split partition $N[v]$ and $I\setminus \{v\}$;
    \\
    \>\> solve $G''$ as case 1, but append $C\setminus N(v)$ to each solution found;
    \\
    5.\> {\bf return} a set in $\cal S$ with the minimum cardinality.
  \end{tabbing}  
    \end{minipage}
  };
  \caption{Algorithm for \problem{split}{unit interval}.}
  \label{fig:split-uig}
\end{figure}
\begin{theorem}\label{split-uig}
  The \problem{split}{unit interval} problem is in $P$.
\end{theorem}
\begin{proof}
  Let $G$ be the input graph to the \problem{split}{unit interval} problem and let $C\uplus I$ be a split partition of $G$.  We use the algorithm in Figure~\ref{fig:split-uig} to find a solution.  To argue its correctness, we show that all sets put into $\cal S$ in steps~1--4 are solutions to $G$, and at least one of them is minimum.  It is clear for step~1.  After the deletion of a solution found in step~2, only $v$ in $I$ remains adjacent to the remaining vertices of $C$.  In step~3, only $v_1$ and $v_2$ from $I$ can remain adjacent to vertices in $C$.   In step~3.2, no vertex in $C$ is adjacent to both $v_1$ and $v_2$; in step~3.3, every vertex in $C$ is adjacent to at least one of $v_1$ and $v_2$.  In either case, it is easy to verify that the graph is a unit interval graph by building a unit interval model directly.  Step~4 follows from the same argument as above: After the deletion of $C\setminus N(v)$, it reduces to one of the three previous steps. 

  Let $V_-$ be a minimum solution to $G$.   If $G - V_-$ is $\{2K_2, P_3\}$-free, then the solution found by step~1 is the minimum.  Henceforth we assume that $G - V_-$ contains a non-clique component $U$; note that such a component contains all vertices in $C\setminus V_-$ and hence is unique.

  In the first case, every vertex $v\in U\cap I$ has at least one non-neighbor in $C\setminus V_-$, i.e., $N(v) \setminus V_-\subset C\setminus V_-$.  According to Proposition~\ref{lem:split+uig}, $|U\cap I| \le 2$.  If $U\cap I = \{v\}$, then $G - (V_- \cup \{v\})$ is $\{2 K_2, P_3\}$-free and the only nontrivial clique $U\setminus\{v\}$ is a subset of $C$; hence step~2 always find a minimum solution.  In the rest of this case, $U\cap I$ has two different vertices; let them be $v_1$ and $v_2$.  Since any $u_1\in N(v_1)\cap N(v_2)$ and $u_2\in C\setminus \big( N(v_1)\cup N(v_2) \big)$ induce a claw with $\{v_1, v_2\}$, at least one of the two sets needs to be empty or completely contained in $V_-$.  Steps~3.2 and 3.3 take care of these two situations separately.

We are now in the second case, where $C\setminus V_-\subseteq N(v)$ for some vertex $v\in I\setminus V_-$; in other words, $V_-$ contains all vertices in $C\setminus N(v)$.  There might be two of such vertices, when we can take $v$ to be either of them.  Clearly,  $N[v]$ and $I\setminus \{v\}$ is then a split partition of $G'' = G - \big(C\setminus N(v) \big)$, which has a solution $V_-\setminus \big(C\setminus N(v) \big)$.  Moreover, under this new split partition, we reduce it to the first case.
 
The algorithm makes $O(n^3)$ calls to the algorithm for the bipartite vertex cover problem, each  taking $O(m\sqrt{n})$ time, and hence the whole algorithm runs in $O(m n^{3.5})$ time.
\end{proof}

We now turn to problems whose inputs are interval graphs, for which we rely on interval models.  Recall that an interval model for an interval graph is a set of intervals representing its vertices.  In this paper, all intervals are closed.  An interval model can be specified by the $2 n$ endpoints for the $n$ intervals, the interval for vertex $v$ being by $[\lp{v}, \rp{v}]$.

For the \problem{unit interval}{complete split} problem, the clique is from some maximal clique of the input graph $G$ and can be enumerated.  On the other hand, according to Proposition~\ref{lem:split+uig}, there are at most three vertices in the independent set, which can be easily found.  However, for interval graphs, it can be more complicated.

\begin{figure}[h]
  \centering\scriptsize
  \begin{tikzpicture}[scale=.6]
    \draw[dashed] (.8, 0) -- (24, 0);
    \draw[dashed] (0.8,4.5) -- (24,4.5);

    \draw[violet, {|[right]}-{|[left]}]  (1.5,1.15)  to (10.5,1.15); 
    \draw[violet, {|[right]}-{|[left]}]  (1.6,1.35)  to (14,1.35); 
    \draw[violet, {|[right]}-{|[left]}]  (1.7,1.55)  to (14.5,1.55); 
    \draw[violet, {|[right]}-{|[left]}]  (3.5,1.75)  to (16.5,1.75); 
    \draw[violet, {|[right]}-{|[left]}]  (5,1.95)  to (17.4,1.95);
    \draw[violet, {|[right]}-{|[left]}]  (8,2.15) to (20.5,2.15);
    \draw[violet, {|[right]}-{|[left]}]  (11.5,1.15) to (21.5,1.15);
    \begin{scope}[xshift= 2.3cm, yshift= .2cm]
     \foreach \n in {1, ..., 5}{
       \draw[thick, violet,{|[right]}-{|[left]}] (\n * .1, \n * .15) to (20 + \n * .1, \n * .15);
     }
     \end{scope}
    
    \begin{scope}[yshift=1.4cm]
      \foreach \x/\y/\l in {1/1/2, 2/2/2.5, 4/1/2.5, 6/2/4, 7./1/1.5, 9.5/1/2.5, 11/2/2.5, 13/1/2.5, 15/2/2.5, 17/1/2.5, 19/2/2.5, 21/1/2.5}{
        \foreach \n in {0, ..., 4}{
          \draw[cyan,{|[right]}-{|[left]}] (\x + \n * .05, \y + \n * .15) to (\x + \l + \n * .05, \y + \n * .15);
        }
      }    
    \draw[thick, cyan, {|[right]}-{|[left]}]  (1, 1) to ["$v_\ell$" below] (3, 1);
    \draw[thick, cyan, {|[right]}-{|[left]}]  (2.2, 2.6) to (4.7, 2.6);
    \draw[thick, cyan, {|[right]}-{|[left]}]  (7., 1) -- (8.5, 1);
    \draw[thick, cyan, {|[right]}-{|[left]}]  (7.2, 1.6) -- (8.7
    , 1.6);
    \draw[thick, cyan, {|[right]}-{|[left]}]  (9.5, 1) to (12, 1);
    \draw[thick, cyan, {|[right]}-{|[left]}]  (11.2, 2.6) -- (13.7, 2.6);
    \draw[thick, cyan, {|[right]}-{|[left]}]  (15, 2) -- (17.5, 2);
    \draw[thick, cyan, {|[right]}-{|[left]}]  (17.2, 1.6) -- (19.7, 1.6);
    \draw[thick, cyan, {|[right]}-{|[left]}]  (21, 1) to ["$v_r$" below] (23.5, 1);
    \draw[thick, cyan, {|[right]}-{|[left]}]  (21.2, 1.6) -- (23.7, 1.6);
     \end{scope}
    \draw[dashed]  (1,2.4) to (1,4.4) node[above] {$\alpha_1$};
    \draw[dashed]  (4.7,4) to (4.7,4.4) node[above] {$\beta_1$};
    \draw[dashed]  (7,2.4) to (7,4.4) node[above] {$\alpha_2$};
    \draw[dashed]  (8.7,3) to (8.7,4.4) node[above] {$\beta_2$};
    \draw[dashed]  (9.5,2.4) to (9.5,4.4) node[above] {$\alpha_3$};
    \draw[dashed]  (13.7,4) to (13.7,4.4) node[above] {$\beta_3$};
    \draw[dashed]  (15,3.4) to (15,4.4) node[above] {$\alpha_4$};
    \draw[dashed]  (19.7,3) to (19.7,4.4) node[above] {$\beta_4$};
    \draw[dashed]  (21,2.4) to (21,4.4) node[above] {$\alpha_5$};
    \draw[dashed]  (23.7,3) to (23.7,4.4) node[above] {$\beta_5$};

    \draw[dashed]  (3,2.4) to (3, 0) node[below] {$\alpha$};
    \draw[dashed]  (21,2.4) to (21, 0) node[below] {$\beta$};

  \end{tikzpicture}
  \caption{Given is an interval model for an interval graph $G$.  \\The top line is for the \problem{interval}{cluster} problem.  A maximum cluster subgraph of $G$ has five cliques, each specified by a pair of $\alpha_i$ and $\beta_i$.  (In this example, the maximum clique in each range $[\alpha_i, \beta_i]$ comprises all intervals in this range.)  \\The bottom  line is for the \problem{interval}{complete split} problem.  A maximum complete split subgraph of $G$ contains 12 vertices.  The clique contains the vertices represented by the lowest five intervals, and the independent set contains $v_\ell$ and $v_r$, together with a maximum independent set of all intervals completely lying in $[\alpha, \beta]$.}
  \label{fig:interval-cluster}
\end{figure}
\begin{theorem}
  Problems \problem{interval}{complete split} and \problem{interval}{cluster} are in P.
\end{theorem}
\begin{proof}
  We solve both problems by finding the maximum subgraphs, for which we work on interval models.  Let us fix an interval model for the input graph $G$; we may assume without loss of generality that no distinct intervals can share an endpoint.

  For the \problem{interval}{complete split} problem, we consider a maximum complete split subgraph $G[U]$.  It is trivial if $G[U]$ is a clique;  hence we assume otherwise.  Let $C\uplus I$ be the split partition of $G[U]$, and let 
\[
\alpha = \rp{v_\ell} = \min_{v\in I}\rp{v}\text{ and }\beta = \lp{v_r} = \max_{v\in I}\lp{v}.
\]
Note that $|I| \ge 2$, as otherwise $G[U]$ is a clique; hence $v_\ell\ne v_r$ and $\alpha < \beta$.   See Figure~\ref{fig:interval-cluster}.  It is easy to see that a  vertex is in $C$ if and only if its interval fully contains $[\alpha, \beta]$; on the other hand, the maximality of $U$ requires us to take all such vertices.  The independent set $I$ would then consists of $v_\ell, v_r$, and a maximum independent set of the subgraph induced by intervals satisfying $\alpha < \lp{v} < \rp{v} < \beta$.  There are $O(n^2)$ pairs of indices to enumerate, and for each pair, both the clique and a maximum independent set can be found in $O(n)$ time.  The whole algorithm runs in $O(n^3)$ time.

  We now consider the \problem{interval}{cluster} problem.  Suppose that $G[U]$ is a maximum cluster subgraph of $G$ and that it has $k$ cliques.  For the $i$th clique $B_i$, we can find two endpoints 
\[
\alpha_i = \min_{v\in B_i} \lp{v} \text{ and } \beta_i = \max_{v\in B_i} \rp{v}.
\]
  Then all intervals for vertices in $B_i$ are completely contained in the interval $[\alpha_i, \beta_i]$.  The $k$ intervals defined as such are pairwise disjoint:  There cannot be edges between two cliques in $G[U]$.  Therefore, $B_i$ must be a maximum clique in the subgraphs induced by $\{v: \alpha_i \le \lp{v} < \rp{v} \le \beta_i\}$, which can be found easily.  See Figure~\ref{fig:interval-cluster}.  The problem can thus be reduced to find the $k$ pairs of endpoints $\alpha_i$ and $\beta_i$.

We build another weighted interval model as follows.  For each $\lp{v_\ell}$ and each $\rp{v_r}$ with $\lp{v_\ell} < \rp{v_r}$, possibly $v_\ell = v_r$, we add an interval $[\lp{v_\ell}, \rp{v_r}]$, whose weight is set to be the size of maximum cliques in the subgraphs induced by $\{v: \lp{v_\ell} \le \lp{v} < \rp{v} \le \rp{v_r}\}$.  We then find a set of pairwise disjoint intervals with the maximum weight sum (or equivalently, a maximum-weight independent set of the weighted interval graph represented by the new interval model).  All the steps can be done in polynomial time. 
\end{proof}

It is easy to verify the following greedy algorithm solves the \problem{tree}{cluster} problem.  We root the input graph at an arbitrary vertex, and work on any leaf at the lowest level:  If it has siblings (i.e., its parent has degree larger than $2$), then delete its parent and put it into the solution; otherwise the parent of its parent.  As we see below, a similar idea would enable us to solve the \problem{block}{cluster} problem.  Recall that a \emph{block} (also known as biconnected component) of a graph $G$ is a maximal biconnected subgraph of $G$.  The \emph{block-cut tree} of a block graph has a vertex for each block and for each cut vertex, and an edge for each pair of a block and a cut vertex that belongs to that block.  Note that every block of a block graph is a clique.
\begin{theorem}\label{thm:interval-cluster}
  The \problem{block}{cluster} problem is in P.
\end{theorem}
\begin{proof}
  We construct the block-cut tree $T$ of the input graph $G$.  A cut vertex $v$ of $G$ is denoted by the same label in $T$, while for a block vertex $u$ of $T$, we use $B(u)$ to denote the vertices in the block of $G$.  We arbitrarily root $T$ at some block vertex.  Note that all leaves of $T$ are block vertices, and their neighbors are not; this invariant will be maintained during our algorithm.  
Until the tree becomes empty, the algorithm always picks a leaf vertex $u$ at the lowest level.   Let $v$ be its parent.  If $v$ has other children, we remove $v$ and its children from $T$ and put $v$ in the solution $V_-$.  In the rest $u$ is the only child of $v$; let $u'$ be the parent of $v$, and let $v'$ be the parent of $u'$.  If at least one vertex in the clique $B(u')$ is not a cut vertex, then we  remove $v, u$ from $T$ and put $v$ in $V_-$.  Otherwise, we remove the subtree rooted at $v'$ from $T$; we put $B(u')\setminus \{v\}$ into the solution, and for each other child $u_i$ of $v'$ that is not a leaf, we solve the subgraph induced by $B(u_i)$ and its children.  The correctness is quite straightforward, so we omit here.
\end{proof}

The last three problems are from chordal graphs.
\begin{theorem}\label{thm:chordal-co-chain}
  The \problem{chordal}{co-chain} problem is in P.
\end{theorem}
\begin{proof}
  The vertices of a {co-chain} graph can be partitioned into two cliques.  On the other hand, any two maximal cliques of a chordal graph together induce a co-chain graph.  Therefore, the problem is to find two maximal cliques with the maximum cardinality together.  Since a chordal graph has at most $n$ maximal cliques, It can be easily calculated in $O(n^2)$ time.
\end{proof}

\begin{theorem}
  For any $p > 1$,  the \problem{chordal}{$K_p$-free} problem is in P.
\end{theorem}
 \begin{proof}
   It is known that a chordal graph is $K_p$-free if and only if it
   has treewidth at most $p - 2$.  Thus the problem is to find an
   induced subgraph of treewidth at most $p - 2$ with the maximum
   number of vertices.  It is known that such a problem can be solved
   in polynomial time for chordal graphs
   \cite{yannakakis-87-k-colorable-subgraph}: Note that a chordal graph is $K_p$-free if and only if it can be colored by $p-1$ colors.
 \end{proof}

\begin{theorem}[\cite{ekim-05-split-coloring}]
  The \problem{chordal}{split} problem is in P. 
\end{theorem}

We remark that Theorem~\ref{thm:split-cluster}--\ref{thm:interval-cluster} can be adapted for the weighted versions of the problems.

\section{Hardness}\label{sec:hardness}

We now turn to hardness results.  Here the problems should be understood to be their decision versions: The input includes, apart from a graph $G$ from ${\cal C}_1$, a positive integer $k$, and the problem is to decide whether $G$ can be made a graph in ${\cal C}_2$ by deleting at most $k$ vertices.  All of them are in NP because all the concerned graph classes can be recognized in polynomial time.
Our first hardness result, on \problem{split}{threshold}, follows easily from the results of Yannakakis~\cite{yannakakis-81-bipartite-node-deletion} on bipartite graphs.  Recall that a bipartite graph is not a chain graph if and only if it contains some $2 K_2$, and a split graph is not a threshold graph if and only if it contains some $P_4$.

\begin{lemma}
  The \problem{split}{threshold} problem is NP-complete.
\end{lemma}
\begin{proof}
  Let $G$ be a bipartite graph with partition $C$ and $I$.  We add all possible edges among $C$ to make it a clique.  Let $G'$ be the
  resulting graph, which is clearly a split graph, witnessed by the split partition $C\uplus I$.  We argue for every
  vertex set $U$ that $G[U]$ is a chain graph, i.e., being $2
  K_2$-free, if and only if $G'[U]$ is a threshold graph, i.e., being
  $P_4$-free.
  Let $X$ be any set of four vertices.  If $G[X]$ is $2 K_2$, then
  $|X\cap C| = |X\cap I| = 2$, but then $G'[X]$ would be a $P_4$.  The
  other direction can be argued similarly.
  Since the \problem{bipartite}{chain} problem is
  NP-hard~\cite{yannakakis-81-bipartite-node-deletion}, the lemma follows.
\end{proof}

Recall that every threshold graph is an interval graph, and this can be generalized as follows.
Let $G_1$ and $G_2$ be two threshold graphs with split partitions $C\uplus I$ and $C'\uplus I'$ respectively.  We let $G_1 \bowtie_{(C, C')} G_2$, or simply $G_1 \bowtie G'$ as in the rest of the paper the partitions are always clear from context, denote the graph obtained from them by adding all possible edges between $C$ and $C'$---i.e., its vertex set and edge set are $V(G_1)\cup V(G_2)$ and $E(G_1)\cup E(G_2)\cup (C\times C' )$ respectively.  This is clearly a split graph with split partition $C \cup C'$ and $I \cup I'$.  One can verify that $G_1 \bowtie G_2$ is also an interval graph by their obstructions as follows.  A split graph that is not an interval graph has to contain a tent, a net, or a rising sun (see Figure~\ref{fig:small-graphs}).   Each of them has three independent vertices, which have to be from $I\cup I'$, but a quick inspection of these three graphs will convince us that this cannot be possible.

\begin{proposition}\label{lem:merge-threshold}
  For any two threshold graphs $G_1$ and $G_2$, the graph $G_1 \bowtie
  G_2$ is an interval graph.
\end{proposition}

A better way to look at Proposition~\ref{lem:merge-threshold} is probably through interval models.\footnote{The following two paragraphs and two figures are for illustration purpose.  They relate the intuition behind our reduction, but are not directly used in the arguments to follow.  The reader may safely skip them if you prefer.}
Let $G$ be a threshold graph with split partition $C\uplus I$, and let
vertices in $I$ be ordered in a way that $N(v_1) \subseteq N(v_2)
\subseteq \cdots \subseteq N(v_{|I|}) $.  We can build an interval
model for $G$ by setting intervals
\begin{align}
   [i, i + 0.5] &&\text{ for every } v_i\in I,
  \notag
  \\
  \label{eq:threshold}
   \big [\min \{i : v_i\in N(v)\}, |I| + 2 \big] &&\text{ for every }
  v\in N(I), \text{ and}
  \\
  \notag
  \big[|I| + 1, |I| + 2 \big] &&\text{ otherwise }(i.e., v\in C\setminus
  N(I)).
\end{align}
See Figure~\ref{fig:threshold-model} for
illustration.

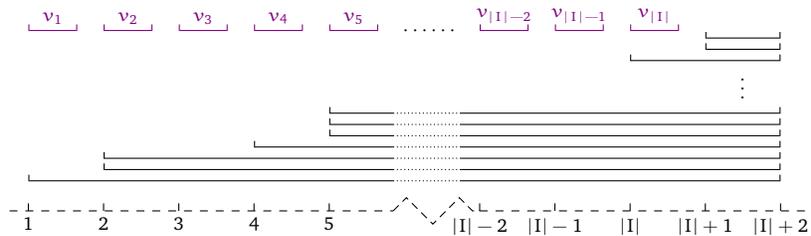
\begin{figure}[h]
  \centering\scriptsize
  \begin{tikzpicture}[scale=.5]
    \draw[dashed] (1.5, 0) -- (11.7,0) (14,0) -- (23,0);
    \draw[densely dashed, decorate,decoration={zigzag, amplitude=5pt, segment length=20pt}] (11.7, 0) -- (14, 0);

    \foreach \x in {1, ..., 5}{
      \draw[violet,{|[right]}-{|[left]}] (\x* 2, 4.8) to ["$v_{\scriptsize\x}$"] (\x*2 + 1.3, 4.8);
      \draw[thin,dashed] (\x*2 , 0) edge (\x*2, 0.2) node[below] {$\x$}; 
    }
    \node at (12.7, 4.8) {$\cdots\cdots$};
    \foreach \x/\n in {7/|I| - 2, 8/|I| - 1, 9/|I|}{
      \draw[violet,{|[right]}-{|[left]}] (\x* 2, 4.8) to ["$v_{\n}$"] (\x*2 + 1.3, 4.8);
      \draw[thin,dashed] (\x*2 , 0) edge (\x*2, 0.2) node[below] {$\n$}; 
    }
    \draw[thin,dashed] (20 , 0) edge (20, 0.2) node[below] {$|I| + 1$}; 
    \draw[thin,dashed] (22 , 0) edge (22, 0.2) node[below] {$|I| + 2$}; 

    \foreach \n/\x in {1/2, 2/4, 3/4, 4/8, 5/10, 6/10, 7/10}{
      \draw[{|[right]}-]  (\x, .5 +\n * 0.3) -- (11.7, .5 + \n * 0.3);
      \draw[densely dotted] (11.7, .5 +\n * 0.3) -- (13.5, .5 +\n * 0.3);
      \draw[-{|[left]}] (13.5, .5 +\n * 0.3) -- (22, .5 +\n * 0.3);
    }
    \node[rotate=90] at (21, 3.3) {\bf$\cdots$};
    \draw[{|[right]}-{|[left]}]  (18,4.) -- (22,4.);
    \draw[{|[right]}-{|[left]}]  (20,4.3) -- (22,4.3);
    \draw[{|[right]}-{|[left]}]  (20,4.6) -- (22,4.6);
  \end{tikzpicture}
  \caption{The interval model for a threshold graph given by \eqref{eq:threshold}.}
  \label{fig:threshold-model}
\end{figure}
An interval model for $G_1\bowtie G_2$ can be built from the interval models for $G_1$ and $G_2$ by (i) keeping the intervals for $G_1$, and (ii) setting the interval to be $\big[ |I|+ |I'| + 3 - \rp{v}, |I|+ |I'| + 3 - \lp{v} \big]$ for each $v\in V(G_2)$.  See Figure~\ref{fig:joint-threshold}.  

\begin{figure}[h]
  \centering
  \begin{tikzpicture}[scale=.4]
  \begin{scope}
    \draw[dashed] (1.5, 0) -- (11.7,0) (14,0) -- (23,0);
    \draw[densely dashed, decorate,decoration={zigzag, amplitude=5pt, segment length=20pt}] (11.7, 0) -- (14, 0);

    \foreach \x in {1, ..., 5}{
      \draw[violet,{|[right]}-{|[left]}] (\x* 2, 4.8) to (\x*2 + 1.3, 4.8);
      \draw[thin,dashed] (\x*2 , 0) edge (\x*2, 0.2) node[below] {\scriptsize$\x$}; 
    }
    \node at (12.7, 4.8) {$\cdots$};
    \foreach \x/\n in {7/|I| - 2, 8/, 9/|I|}{
      \draw[violet,{|[right]}-{|[left]}] (\x* 2, 4.8) to (\x*2 + 1.3, 4.8);
      \draw[thin,dashed] (\x*2 , 0) edge (\x*2, 0.2) node[below] {\scriptsize$\n$}; 
    }
    \draw[thin,dashed] (20 , 0) edge (20, 0.2); 
    \draw[thin,dashed] (22 , 0) edge (22, 0.2) node[below] {\scriptsize$|I| + 2$}; 

    \foreach \n/\x in {1/2, 2/4, 3/4, 4/8, 5/10, 6/10, 7/10}{
      \draw[{|[right]}-]  (\x, .5 +\n * 0.3) -- (11.7, .5 + \n * 0.3);
      \draw[densely dotted] (11.7, .5 +\n * 0.3) -- (13.5, .5 +\n * 0.3);
      \draw[-{|[left]}] (13.5, .5 +\n * 0.3) -- (22, .5 +\n * 0.3);
    }
    \node[rotate=90] at (21, 3.3) {\bf$\cdots$};
    \draw[{|[right]}-{|[left]}]  (18,4.) -- (22,4.);
    \draw[{|[right]}-{|[left]}]  (20,4.3) -- (22,4.3);
    \draw[{|[right]}-{|[left]}]  (20,4.6) -- (22,4.6);
  \end{scope}
    
  \begin{scope}[rotate=180, xshift=-42cm]
    \draw[dashed] (1.5, 0) -- (11.7,0) (14,0) -- (23,0);
    \draw[densely dashed, decorate,decoration={zigzag, amplitude=5pt, segment length=20pt}] (11.7, 0) -- (14, 0);

    \foreach \x in {1, ..., 5}{
      \draw[violet,{|[left]}-{|[right]}] (\x* 2, 4.8) to (\x*2 + 1.3, 4.8);
    }
      \draw[thin,dashed] (2 , 0) edge (2, -0.2) node[below] {\scriptsize$|I| + |I'| + 2$}; 
    \node at (12.7, 4.8) {$\cdots$};
    \foreach \x/\n in {7/|I| - 2, 8/|I| - 1, 9/|I|}{
      \draw[violet,{|[left]}-{|[right]}] (\x* 2, 4.8) to (\x*2 + 1.3, 4.8);
    }

    \foreach \n/\x in {1/2, 2/4, 3/4, 4/8, 5/10, 6/10, 7/10}{
      \draw[{|[left]}-]  (\x, .5 +\n * 0.3) -- (11.7, .5 + \n * 0.3);
      \draw[densely dotted] (11.7, .5 +\n * 0.3) -- (13.5, .5 +\n * 0.3);
      \draw[-{|[right]}] (13.5, .5 +\n * 0.3) -- (22, .5 +\n * 0.3);
    }
    \node[rotate=90] at (21, 3.3) {\bf$\cdots$};
    \draw[{|[left]}-{|[right]}]  (18,4.) -- (22,4.);
    \draw[{|[left]}-{|[right]}]  (20,4.3) -- (22,4.3);
    \draw[{|[left]}-{|[right]}]  (20,4.6) -- (22,4.6);
  \end{scope}
    
  \end{tikzpicture}
  \caption{The interval model for $G_1\bowtie G_2$.}
  \label{fig:joint-threshold}
\end{figure}
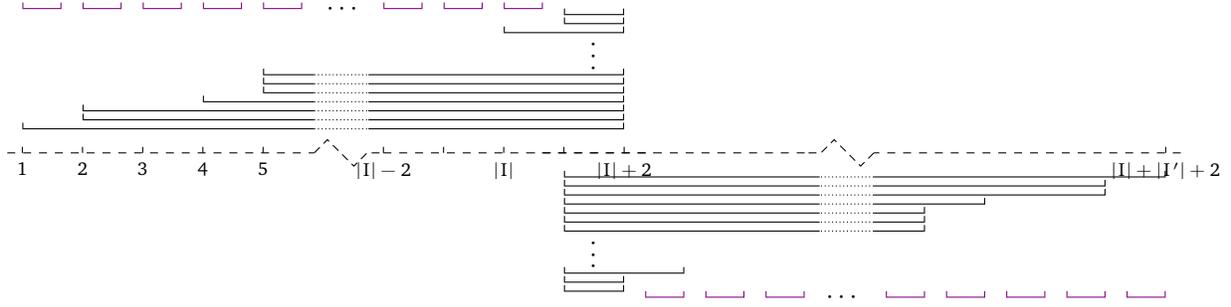

We are now ready to prove the first major theorem of this section.  
\begin{theorem}\label{thm:chordal-to-interval}
  The \problem{split}{interval} problem is NP-complete.
\end{theorem}
\begin{proof}
It is clear that the problem is in NP.  Let $G$ be a split graph with split partition $C\uplus I$.  We take a complete split graph $G'$ with split partition $C'\uplus I'$, where $|C'| = |I'| = |C|$, and let $H = G\bowtie G'$.
  We argue that ($G, k$) is a yes-instance of the \problem{split}{threshold} problem if and only if ($H, k$) is a yes-instance of the \problem{split}{interval} problem.
Since both problems are trivial yes-instances when $k \ge |C|$, we may assume henceforth $k < |C|$.  
  
  Suppose that $G - V_-$, where $|V_-| \le k$, is a threshold graph.  According to Proposition~\ref{lem:merge-threshold}, $( G - V_-) \bowtie G'$ is an interval graph.  It is the same graph as $H - V_-$.  Therefore, $V_-$ is a solution of ($H, k$).  This verifies the only if direction.

  Now suppose that $H - V_-$, where $|V_-| \le k$, is an interval graph.  Suppose for contradiction that $\underline G = G - \big( V_-\cap V(G) \big)$ is not a threshold graph.  Then $\underline G$ must contain some $P_4$; let it be $v_1 u_1 u_2 v_2$.  Since $\underline G$ is a split graph, we must have $u_1, u_2\in C$ and $v_1, v_2\in I$.  On the other hand, by the assumption $k < |C|$, neither $C'\setminus V_-$ nor $I'\setminus V_-$ can be empty.  Let $u\in C'\setminus V_-$ and $v\in I'\setminus V_-$.  By the construction, the only edges between $\{u, v\}$ and $\{v_1, u_1, u_2, v_2\}$ are $u u_1$ and $u u_2$,  but then these six vertices together induce a net in $H - V_-$, a contradiction.
\end{proof}

\begin{corollary}
  The \problem{chordal}{interval} problem is NP-complete.
\end{corollary}

The last result is on the deletion of any biconnected subgraph from chordal graphs.  Recall that a vertex $v$ is \emph{simplicial} in $G$ if $N[v]$ is a clique.  A graph is chordal if and only if we can make it empty by deleting simplicial vertices in the remaining graph \cite{dirac-61-chordal-graphs}.
\begin{theorem}\label{thm:biconnected-2}
  Let $F$ be a biconnected chordal graph.  If $F$ is not complete, then the \problem{chordal}{$F$-free} problem is NP-complete.  Moreover, if $F$ is a complete split graph with $|C| = 2$ and $|I| \ge 2$, then the \problem{split}{$F$-free} problem is NP-complete.
\end{theorem}
\begin{proof}
  We use the following reduction from the vertex cover problem.  Let $G$ be an input graph to the vertex cover problem, we conduct the following operations.
  \begin{enumerate}
  \item For each edge $uv \in E(G)$, add a distinct copy of $F$ such that each of them uses $u v$ as one of its edges.  We say that $u, v$ are the attachments for this copy of $F$.
  \item Add all possible edges among $V(G)$ to make it complete.
  \end{enumerate}
  Let $G'$ be the obtained graph.  To see that $G'$ is chordal, we give an explicit way of eliminating simplicial vertices to make $G'$ empty.  A chordal graph either is a clique or contains two nonadjacent simplicial vertices; all vertices are simplicial when it is a clique.  For each copy of $F$, we can find a simplicial vertex in $V(F)\setminus \{u, v\}$.  We keep doing this, and then only vertices in $V(G)$ remain.  They have been made a clique, and thus all of them simplicial.

  We argue that $G$ has a vertex cover of size $k$ if and only if we can delete $k$ vertices from $G'$ to make it $F$-free.  The following fact would be essential.  We consider any copy $X$ of $F$ with attachments $u$ and $v$.  If we delete $u$ or $v$, then the other becomes a cut vertex, and $X\setminus \{u, v\}$ are in different blocks from other vertices of $V(G')$.  But any other copy of $F$, if it exists, must be completely contained in a block, and thus it cannot contains any vertex in $X$.

Suppose that $V_-$ is a vertex cover of size $k$ in $G$.  We claim that $\underline G = G' - V_-$ has no copy of $F$.  For each copy of $F$ with attachments $u$ and $v$.  Therefore, a copy of $F$ in $\underline G$, if one exists, has all its vertices from $V(G)$.{  But this is not possible because $F$ is not a clique.}

Suppose  now that $V_-$ is a solution to $G'$ of size $k$.  We may assume that $V_-$ contains no new vertex: If it contains a vertex from a copy of $F$ with attachments $u$ and $v$, we can replace it by $u$.  (Note that the new set remains a solution to $G'$ because the aforementioned fact.)  Since $G' - V_-$ does not contain $F$, each copy of it has at least one of the attachments in $V_-$.  Therefore, each edge of $G$, at lest one end is in $V_-$, which means that $V_-$ is a vertex cover of $G$.
\end{proof}

\begin{figure}[h]
  \centering
  \begin{tikzpicture}[auto=left,
    every node/.style={vertex}, 
    every path/.style={thick}, scale=.6]
    
    \node["$v_1$" xshift=4mm] (v1) at (2,2) {};
    \node["$v_2$" xshift=-4mm] (v2) at (-2,2) {};
    \node["$v_3$" xshift=-4mm] (v3) at (-2,-2) {};
    \node["$v_4$" xshift=4mm] (v4) at (2,-2) {};
    \draw (v2)-- (v3)-- (v4)--(v1)-- (v2);
    \draw [gray, thin] (v2)-- (v4) (v3)-- (v1);

    \foreach \x/\y/\a in {1/2, 2/3, 3/4, 4/1} {
      \begin{scope}[rotate = 90 * \x - 90]
        \node[gray] (c\x\y) at (0, 1) {};
        \node[gray] (u2\x\y) at (0, 3) {};
        \node[gray] (u3\x\y) at (2,3.5) {};
        \node[gray] (u1\x\y) at (-2,3.5) {};
        \draw [thin, cyan] (c\x\y)--(v\x) (c\x\y)--(v\y)
        (u2\x\y)--(v\x) (u2\x\y)--(v\y)
        (u3\x\y)--(v\x) (u3\x\y)--(u2\x\y)
        (u1\x\y)--(v\y) (u1\x\y)--(u2\x\y);
      \end{scope}
    }
  \end{tikzpicture}
  \caption{Reduction for Theorem~\ref{thm:biconnected-2}, with $F$ being a tent.  The original graph $G$, drawn with blue vertices and thick edges, is a $C_4$.  The new vertices are gray and new edges thin.  The set $\{v_1, v_3\}$ is a solution to both problems.}
  \label{fig:net-free}
\end{figure}
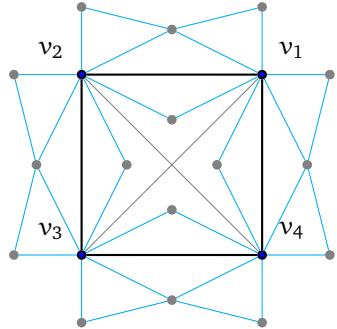

\begin{corollary}
  Problems \problem{split}{block} and \problem{chordal}{block} are NP-complete.
\end{corollary}

\appendix
\section*{Appendix}
\tikzstyle{vertex}  = [{fill=olive,olive,circle,draw,inner sep=1pt}]
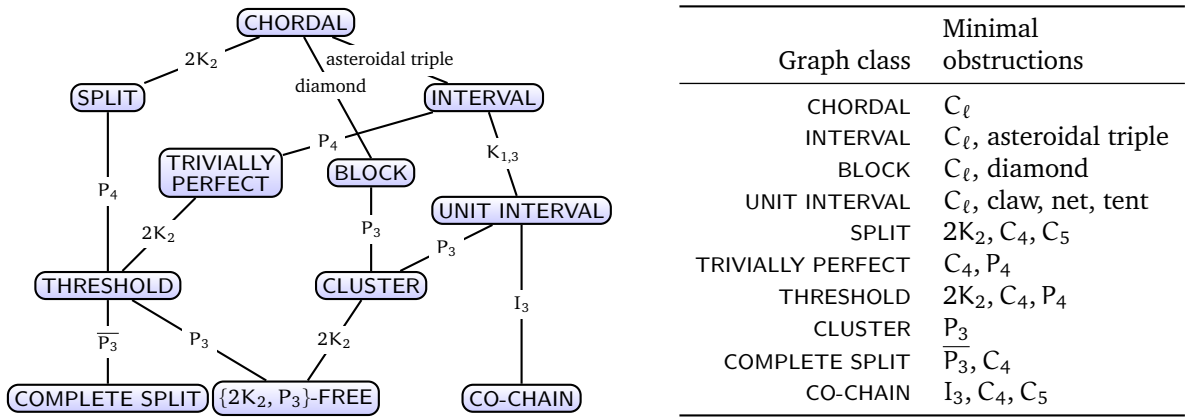
\begin{figure}[h]
  \centering
  \subfloat{\label{fig:containment}
  \begin{tikzpicture}[every path/.style={thick}, baseline=(uig.base)]
    \scriptsize
    \node[class] (chordal) at (2,5)  {\hsc{chordal}};                                 
    \node[class] (split) at (-0.5,4)             {\hsc{split}};                                   
    \node[class] (interval) at (4.5,4)           {\hsc{interval}};                                
    \node[class] (threshold) at (-0.5,1.5)       {\hsc{threshold}};                               
    \node[class] (tpg) at (1,3)                  {\hsc{trivially}\\\hsc{perfect}};                
    \node[class] (uig) at (5,2.5)                {\hsc{unit interval}};                           
    \node[class] (block) at (3,3)                {\hsc{block}};                                   
    \node[class] (cluster) at (3,1.5)            {\hsc{cluster}};                                 
    \node[class] (csplit) at (-0.5,0)            {\hsc{complete split}};                          
    \node[class] (bottom) at (2,0)               {\hsc{$\{2K_2, P_3\}$-free}};                    
    \node[class] (cochain) at (5,0)              {\hsc{co-chain}};                                
    \draw (csplit) -- (threshold) node [midway,fill=white] {\scriptsize $\overline{P_3}$};
    \draw (threshold) -- (split) node [midway,fill=white] {\scriptsize $P_4$};
    \draw (split) -- (chordal) node [midway,fill=white] {$2K_2$};
    \draw (interval) -- (chordal) node [midway,fill=white] {\scriptsize asteroidal triple};
    \draw (uig) -- (interval) node [midway,fill=white] {\scriptsize $K_{1,3}$};
    \draw (cluster) -- (uig) node [midway,fill=white] {\scriptsize $P_{3}$};
    \draw (cluster) -- (block) node [midway,fill=white] {\scriptsize $P_{3}$};
    \draw (block.north) -- (chordal) node [pos=.6,fill=white] {\scriptsize diamond};
    \draw (bottom) -- (cluster) node [midway,fill=white] {\scriptsize $2 K_{2}$};
    \draw (bottom) -- (threshold) node [midway,fill=white] {\scriptsize $P_{3}$};
    \draw (threshold) -- (tpg) node [midway,fill=white] {\scriptsize $2 K_{2}$};
    \draw (tpg) -- (interval) node [pos=.3, fill=white] {\scriptsize $P_{4}$};
    \draw (cochain) -- (uig) node [midway,fill=white] {\scriptsize ${I_3}$};
  \end{tikzpicture}
  }
  \qquad
  \subfloat{\label{fig:obstructions}
  \begin{tabular}{r l}
    \toprule
    & Minimal
    \\
    Graph class & obstructions
    \\
    \midrule
    \hsc{chordal} & $C_\ell$
    \\
    \hsc{interval} & $C_\ell$, asteroidal triple
    \\
    \hsc{block} & $C_\ell$, diamond
    \\
    \hsc{unit interval} & $C_\ell$, claw, net, tent
    \\
    \hsc{split} & $2 K_2, C_4, C_5$
    \\
    \hsc{trivially perfect} & $C_4, P_4$
    \\
    \hsc{threshold} & $2 K_2, C_4, P_4$
    \\
    \hsc{cluster} & $P_3$
    \\
    \hsc{complete split} & $\overline{P_3}, C_4$
    \\
    \hsc{co-chain} & ${I_3}, C_4, C_5$
    \\
    \bottomrule 
  \end{tabular}
}
  \caption{Forbidden induced subgraphs and containment relationships of related graph classes}
  \label{fig:classes-containment}
\end{figure}

The minimal forbidden induced subgraphs for chordal graphs are well
known.  For all the classes at lower levels, their forbidden induced
subgraphs with respect to its immediate super-classes are given on the
edges.  From them we are able to derive all the minimal forbidden
induced subgraphs for each of these classes.
For example, the characterization of unit interval graphs follows from
the characterization of interval graphs and that we can find a claw in
a chordal witness for an \emph{asteroidal triple} (i.e., three
vertices such that each pair of them is connected by a path avoiding
neighbors of the third one) that is not a net or tent.  Likewise, the
minimal forbidden induced subgraphs of trivially perfect graphs can be
derived from those of interval graphs and that all chordal witnesses
for asteroidal triples and all holes that are not $C_4$'s contain a
$P_4$.

\end{document}